%% file: PRR_main_v2.tex
\documentclass[aps,prresearch,reprint,superscriptaddress]{revtex4-2}
\usepackage{amsmath,amsfonts,amssymb}
\usepackage{mathtools}
\usepackage{multirow}
\usepackage{ulem}

\usepackage[utf8]{inputenc}

\usepackage[T1]{fontenc}

\usepackage{bm}
\usepackage{times}

\input{pretex}

\definecolor{beamer}{rgb}{0.2,0.2,0.7}
\definecolor{colorone}{rgb}{1,0.36,0.03}
\definecolor{colortwo}{rgb}{0.4,0.77,0.17}
\definecolor{colorthree}{rgb}{0.01,0.51,0.93}
\definecolor{colorfour}{rgb}{0.47,0.26,0.58}
\definecolor{colorfive}{rgb}{0.12,0.55,0.16}
\usepackage{tcolorbox}
\usepackage{relsize}
\usepackage{graphicx}
\usepackage{booktabs}
\usepackage{amsmath}
\usepackage{tikz}
\usetikzlibrary{quantikz}
\nc{\st}{\text{subject to} \ }
\nc{\supre}{\text{supremum} \ }
\nc{\sdp}{\text{sdp}}
\usepackage{array}

\newcommand{\update}[1]{\textcolor{black}{#1}}
\newcommand{\reupdate}[1]{\textcolor{black}{#1}}
\allowdisplaybreaks

\nc{\ith}[1]{{#1}^\mathrm{th}}

\begin{document}
\title{Probabilistic channel simulation using coherence}
 
\author{Benchi Zhao}
\email{benchizhao@gmail.com}
\affiliation{Graduate School of Engineering Science, Osaka University, 1-3 Machikaneyama, Toyonaka, Osaka 560-8531, Japan}

\author{Kosuke Ito}
\email{kosuke.ito.qiqb@osaka-u.ac.jp}
\affiliation{Center for Quantum Information and Quantum Biology, Osaka University, 1-2 Machikaneyama, Toyonaka 560-8531, Japan}

\author{Keisuke Fujii}
\email{fujii@qc.ee.es.osaka-u.ac.jp}
\affiliation{Graduate School of Engineering Science, Osaka University, 1-3 Machikaneyama, Toyonaka, Osaka 560-8531, Japan}
\affiliation{Center for Quantum Information and Quantum Biology, Osaka University, 1-2 Machikaneyama, Toyonaka 560-8531, Japan}
\affiliation{RIKEN Center for Quantum Computing (RQC), Hirosawa 2-1, Wako, Saitama 351-0198, Japan}

\begin{abstract}
Channel simulation using coherence, which refers to realizing a target channel with coherent states and free operations, is a fundamental problem in the quantum resource theory of coherence. The limitations of the accuracy of deterministic channel simulation motivate us to consider the more general probabilistic framework. In this paper, we develop the framework for probabilistic channel simulation using coherence with free operations. When the chosen set of free operations is the maximally incoherent operations, we provide an efficiently computable semidefinite program (SDP) to calculate the maximal success probability and derive the analytic expression of success probability for some special cases. When the chosen set of free operations is the dephasing-covariant incoherent operations (DIO), we show that if the target channel is not a resource nonactivating channel, then one cannot simulate it exactly both deterministically and probabilistically. The SDP for maximal success probability of simulating a channel by DIO is also given correspondingly.

\end{abstract}

\maketitle

\section{Introduction}
In recent years, many efforts have contributed to establishing the framework of quantum resource theories~\cite{chitambar2019quantum, horodecki2013quantumness} to understand the unique properties of quantum mechanical systems such as coherence~\cite{streltsov2017colloquium, aberg2006quantifying, fang2018probabilistic, regula2018one, zhao2018one, napoli2016robustness, baumgratz2014quantifying, chitambar2016comparison, piani2016robustness, kelly2023coherence,tajima2024gibbs}, entanglement~\cite{horodecki2009quantum, vedral1997quantifying, wang2020cost, chen2023near,zhu2023estimate}, and magic~\cite{howard2017application, bravyi2016trading, wang2019quantifying, wang2020efficiently, bravyi2012magic, zhu2023limitations}. \update{In general, a resource theory is defined by specifying free states and free operations}. Free states are states that do not possess the resource under consideration, while free operations are operations that preserve \update{the set of free states}. Taking the resource theory of entanglement as an example~\cite{horodecki2009quantum, vedral1997quantifying, wang2020cost, chen2023near,zhu2023estimate}, the free states are separable states, which are not entangled, and one of the free operation sets is local operation and classical communication (LOCC)~\cite{chitambar2014everything, nielsen2010quantum}, which does not generate entanglement. Similar to entanglement, coherence is another important topic in quantum resource theories~\cite{chitambar2019quantum}, which refers to the property of the superposition of states. It empowers various quantum tasks, such as cryptography~\cite{coles2016numerical}, metrology~\cite{giovannetti2011advances, frowis2011stable, toth2014quantum}, thermodynamics~\cite{lostaglio2015description, brandao2013resource, gour2015resource}, and channel simulation~\cite{dana2017resource,diaz2018using}.

In the resource theory of coherence~\cite{streltsov2017colloquium}, the free states are defined as classical states, i.e., density operators that are diagonal in a given reference orthogonal basis $\{\ket{i}\}$. Such states are called \textit{incoherent states} and denoted as $\cI$. The corresponding \textit{maximally coherent state} in dimension $m$ is the state $\ket{\Psi_m}=\frac{1}{\sqrt{m}}\sum_{j=0}^{m-1}\ket{j}$. In this work, we denote the density matrix of a maximally coherent state with rank $m$ as $\Psi_m = \proj{\Psi_m}$ for convenience. 
\reupdate{The resource theory of coherence does not have a gold-standard physics-motivated class of operations like LOCC in the entanglement resource theory.}
\reupdate{With this in mind, our task is}
to characterize the operational properties and applications of quantum coherence under several different sets of operations, such as dephasing-covariant incoherent operations (DIO)~\cite{chitambar2016critical}, and maximally incoherent operations (MIO)~\cite{aberg2006quantifying}. \reupdate{An operation $\cM$ is in DIO if it commutes with $\Delta$}, or equivalently $\cM(\proj{i})\in\cI$ and $\Delta(\cM(\ketbra{i}{j}))=0,~\forall i\neq j$, where $\Delta$ is the coherence destroying map~\cite{liu2017resource} (completely dephasing channel), i.e., $\Delta(\cdot)=\sum_i \proj{i}\cdot\proj{i}$.
MIO \update{consists of} all operations $\cM$ such that $\cM(\rho)\in\cI$ for any free state $\rho\in\cI$. From the definition, DIO is a subset of MIO, i.e., DIO $\subset$ MIO.

\begin{figure}[ht]
\centering
\includegraphics[width=\linewidth]{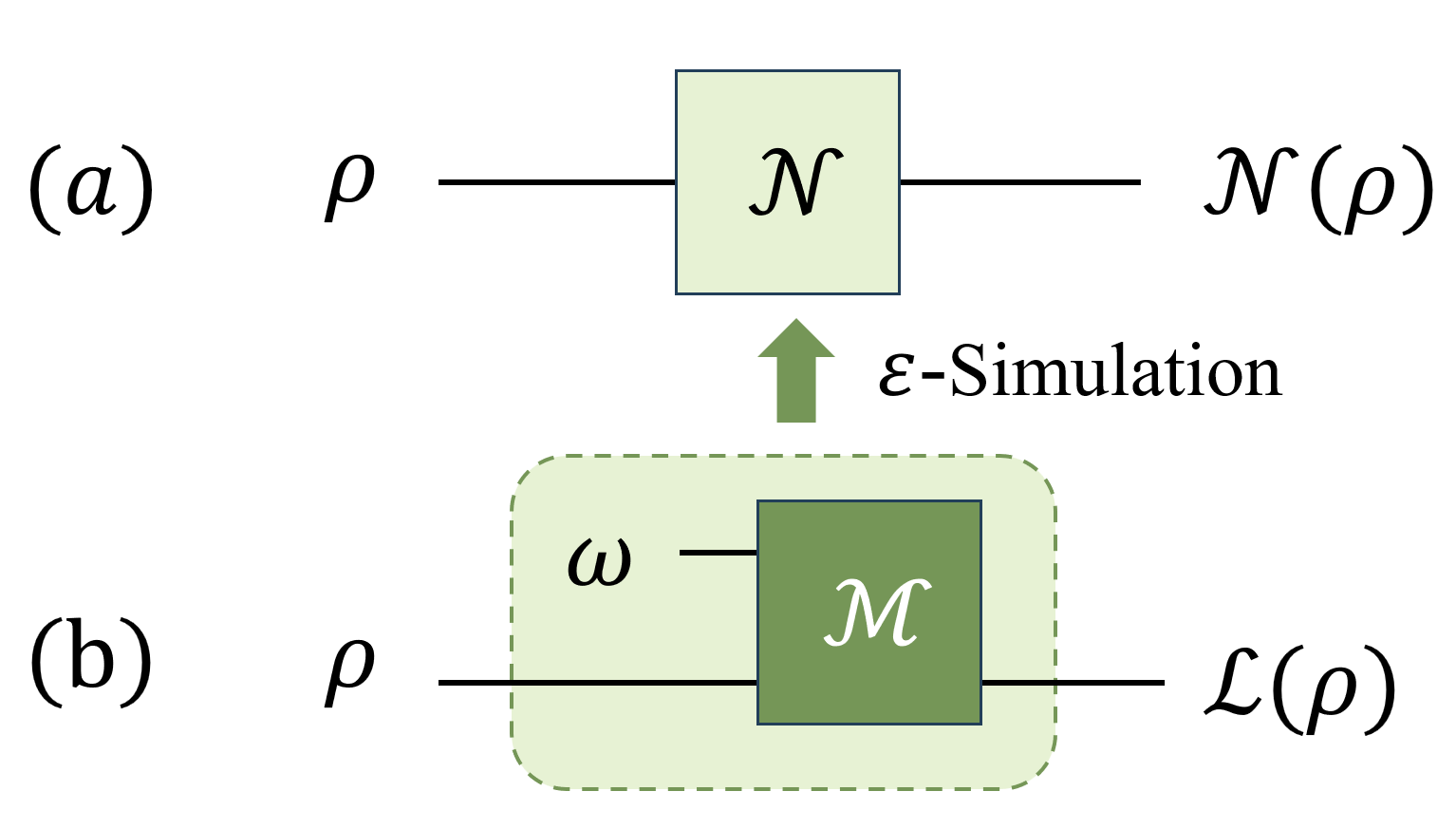}
    \caption{(a) \update{The quantum} channel $\cN$. (b) Utilizing \update{a} free operation $\cM$ and \update{a} given resource state $\omega$ to implement channel $\cL$. The realized channel $\cL$ is $\varepsilon$ close to the target channels $\cN$.}
    \label{fig:Channel_simulation}
\end{figure}

Initial research primarily focused on the quantification and interconversion of \textit{static coherence}~\cite{aberg2006quantifying,baumgratz2014quantifying,winter2016operational}, which refers to the degree of superposition within a state. \update{Later on, a more general framework was proposed, called \textit{dynamic coherence}--the power to generate coherence itself~\cite{diaz2018using}. How to convert static coherence into dynamic coherence (also known as \textit{channel simulation using coherence}, as shown in Fig.~\ref{fig:Channel_simulation}) raises great interest, and many efforts have contributed to establishing the framework of deterministic channel simulation~\cite{dana2017resource, diaz2018using,jones2023hadamard}}. At the same time, the limitations of deterministic channel simulation are also shown.

A given resource state does not necessarily allow for the exact and deterministic MIO simulation of a given target quantum channel, although for each channel there exists a state such that such a simulation is possible~\cite{dana2017resource, diaz2018using}. In fact, if we can only access the resource state $\omega$ of which the robustness of coherence (ROC) is smaller than the requirement of the channel simulation, then it is impossible to simulate the target channel with no error. Also, in Ref.~\cite{jones2023hadamard}, the authors found that any coherent unitary channel cannot be simulated by DIO exactly (e.g., the Hadamard gate) no matter what resource state $\omega$ is provided. More generally, in Ref.~\cite{theurer2019quantifying}, the authors showed that if the target channel is not a resource nonactivating channel~\cite{liu2017resource}, then it cannot be simulated by DIO exactly and deterministically.

The limitation of deterministic channel simulation motivates us to consider a more general probabilistic framework, in which the channel simulation will succeed only with some probability. The probabilistic framework has been applied to many quantum tasks and has shown advantages over deterministic ones. For example, in Ref.~\cite{regula2022probabilistic}, the author observes that one can transfer the quantum state from $\rho$ to $\sigma$ probabilistically, while such transformation is forbidden in the deterministic scenario. Moreover, in the task of coherence distillation~\cite{fang2018probabilistic}, with the same input state, the output state by probabilistic distillation is closer to the maximal coherent state than deterministic distillation.

In this paper, we focus on the probabilistic channel simulation using coherence, characterizing the relation between the maximal success probability and the distance from simulated channel $\cL$ and target channel $\cN$. In the first part, we show three cases of channel simulation with MIO depending on the types of the resource state. (i) If the resource state is the rank-$m$ maximally coherent state $\omega=\Psi_m$, we provide the analytical expression of maximal success probability. (ii) If the resource state $\omega$ is a pure coherent state, we derive the non-zero success probability. (iii) If the resource state $\omega$ is a general coherent state, we provide an efficiently computable semidefinite program (SDP) to achieve the maximal success probability. In the second part, we concentrate on the channel simulation with DIO. \update{It was proved in Ref.~\cite{theurer2019quantifying} that if the target channel $\cN$ is not a resource nonactivating channel, then it cannot be simulated exactly by DIO with any resource state deterministically. We further \update{show} that such an exact simulation is impossible even in probabilistic scenarios.} We then provide the efficiently computable SDP for achieving the maximal success probability of channel simulation with DIO. Our paper fills an important gap in the literature by establishing the probabilistic toolbox for the key resource of quantum coherence.

\section{The problem of probabilistic channel simulation}

To quantify the coherence of a quantum channel, the (ROC) of a quantum channel was proposed 
\begin{definition}{\rm ~\cite{diaz2018using}}\label{def:robustness_of_channel}
    The robustness of coherence of a quantum channel $\cN$, $C_R(\cN)$ is defined by 
    \begin{equation}
        1 + C_R(\cN) := \min\{\lambda:~\cN\le\lambda \cM, \cM \in {\rm MIO}\}.
    \end{equation}
    The inequality of $\cN\le\lambda\cM$ means that the map $\lambda\cM-\cN$ is completely positive.
\end{definition}

\begin{definition}{\rm ~\cite{diaz2018using}}\label{def:robustness_of_channel_smooth}
    The smoothed version of the robustness of a channel is called the $\varepsilon$ robustness of the coherence of channel, which is defined by 
    \begin{equation}
        C_R^\varepsilon(\cN) := \min\Big\{C_R(\cL):~\frac{1}{2}\|\cN-\cL\|_\diamond\le\varepsilon\Big\}.
    \end{equation}
\end{definition}

If we consider replacement channel $\cN_\sigma$, the output of which is independent of the input state $\rho$, i.e., $\cN_\sigma(\rho) = \tr[\rho]\sigma$, then the \update{ROC} of a channel (Definition~\ref{def:robustness_of_channel}) reduces to the \update{ROC} of states~\cite{napoli2016robustness, piani2016robustness}. \update{More details of ROC of states are shown in Appendix ~\ref{appen:Robustness of a state}.}
\update{
\begin{definition}{\rm ~\cite{napoli2016robustness, piani2016robustness}}
    Given an arbitrary quantum state $\rho$, the robustness of coherence of the state is defined as
    \begin{equation}
        C_R(\rho)=\min_{\tau}\Big\{s\ge 0 \Big| \frac{\rho+s\tau}{1+s}=\sigma \in \cI\Big\},
    \end{equation}
    where $\cI$ refers to the set of free states.
\end{definition}
}
Similar to the \update{ROC of a state}, the \update{ROC of a channel} quantifies the minimal mixing required to destroy all the coherence in a quantum channel $\cN$~\cite{diaz2018using}.  \update{The ROC} of quantum channel $C_R(\cN)$, as well as \update{$\varepsilon$-ROC of a channel} can be effectively calculated by SDP~\cite{diaz2018using}, which are shown in Appendix~\ref{appen:SDP}.

The task of probabilistic channel simulation can be defined as follows. For a given target channel $\cN_{A\rightarrow B}$, \update{which transforms linear operators in the system $A$ to the system $B$,} resource state $\omega$ and error tolerance $\varepsilon$, \update{we aim to find a free operation $\cM$ such that $\cM(\omega\ox \cdot)$ probabilistically outputs a channel $\cL$ which is close to $\cN(\cdot)$ up to error $\varepsilon$}. A single-bit classical flag register $F$ is used to indicate if the simulation is successful. If the register $F$ is found to be $\ket{0}$, then it implies the output \update{channel} $\cL(\cdot)$ simulates our target $\cN(\cdot)$ up to an error $\varepsilon$. Otherwise, the simulation fails and we discard the "rubbish" output $\cK(\cdot)$, where $\cK$ can be arbitrary MIO. When the channel simulation fails, we can perform a replacement channel $\cN_{\frac{\mathbb{I}}{d}}$, which replaces input state into identity, on $\cK(\cdot)$ as post-operation, i.e., $\cN_{\frac{\mathbb{I}}{d}}\circ\cK(\cdot) = \tr[\cdot]\mathbb{I}/d$, where $\mathbb{I}$ is identity and $d$ is the dimension of the system. Note that such post-operation will not change the success probability. The state version of this trick is shown in Ref.~\cite{fang2018probabilistic}. Our goal is to maximize the success probability of channel simulation. Here, we can define the problem as follows

\begin{problem}~\label{prob:1}
    Given triplet $(\cN, \omega, \varepsilon)$, what is the maximal success probability $P_\cO(\cN, \omega, \varepsilon)$ to simulate channel $\cN$ up to error $\varepsilon$ with the given resource state $\omega$ and free operation class $\cO\in\{{\rm MIO, DIO}\}$? Mathematically, \update{it is an optimization problem}
    \begin{subequations}\label{eq:original_opt_prob}
    \begin{align}
        P_\cO(\cN, \omega, \varepsilon) &= \max\quad  p;\\
        {\rm s.t.}\quad & \cM_{RA\rightarrow FB}(\omega\ox\cdot)=p\proj{0}_F\ox\cL(\cdot)\nonumber\\
        &\qquad\qquad+ (1-p)\proj{1}_F\ox\tr[\cdot]\frac{\mathbb{I}}{d}; \label{eq:original_opt_prob_2b}\\
        & \frac{1}{2} ||\cL-\cN||_\diamond \le \varepsilon;\\
        &\cM\in \cO,
    \end{align}
    \end{subequations}
    where $\cL$ are the probabilistically implemented channels, which approximates the target channel, and $\mathbb{I}$ refers to identity.
\end{problem}
\reupdate{System $R$ is the resource system containing the resource state}, and system $F$ is the flag system, indicating if the channel simulation is successful. $||\cdot||_\diamond$ is known as the diamond norm~\cite{watrous2009semidefinite}, which has two operational meanings: First, it quantifies how well one physically discriminates between two quantum channels~\cite{sacchi2005optimal}. If we set $\varepsilon=0$, it implies the exact implementation of the target channel. Second, it quantifies the cost for simulating a general hermitian preserving map with physical implementations~\cite{zhao2023power, regula2021operational}.

\section{Probabilistic channel simulation with MIO}

In this section, we are going to study the channel simulation with MIO under three cases. First, if the resource state is the rank-$m$ maximally coherent state, we provide the analytical expression of the maximal success probability. Second, if the resource state $\omega$ is a pure coherent state, we derive analytical lower bounds on the maximal success probability. Third, for general coherent states, we provide an efficiently computable SDP to determine the maximal success probability.

\begin{theorem}\label{theo:MIO_simulatable}
    \update{Given an arbitrary target channel $\cN$ and a fixed maximally coherent state $\Psi_m$ of dimension $m\ge 2$, the maximal success probability of the exact channel simulation with MIO is}
    \begin{equation}\label{eq:theorem_1}
        P_{\rm MIO}(\cN, \Psi_m, \varepsilon=0) = \min\Big\{1, \frac{m-1}{C_R(\cN)}\Big\}.
    \end{equation}
\end{theorem}

\begin{proof}
    If a target channel $\cN$ can be simulated by MIO and the given resource state $\Psi_m$, it means that there exists an operation $\cM\in {\rm MIO}$, such that
    \begin{equation}
        \cM(\Psi_m\ox\cdot) = \cL(\cdot) = p\proj{0}\ox\cN(\cdot) + (1-p)\proj{1}\ox\cK(\cdot),\nonumber
    \end{equation}
    with $p>0$. The output $\cK(\cdot)$ is a failure branch, which should be discarded. And the channel $\cK$ can be arbitrary MIO, i.e., $C_R(\cK)=0$. 
    The main idea is to prove $C_R(\cL) = pC_R(\cN)$, which will be obtained directly by proving $C_R(\cL) \le pC_R(\cN)$ and $C_R(\cL) \ge pC_R(\cN)$. 
    
    \reupdate{Due to the convexity of the ROC of channels~\cite{diaz2018using}, we have
    \begin{align}
        C_R(\cL) &= C_R(p\ketbra{0}{0}\ox\cN + (1-p)\ketbra{1}{1}\ox\cK)\nonumber\\
        &\le pC_R(\ketbra{0}{0}\ox\cN) + (1-p)C_R(\ketbra{1}{1}\ox\cK)\nonumber\\
        &= pC_R(\cN) + (1-p)C_R(\cK)\nonumber\\
        &= pC_R(\cN)\label{eq:theo1_convexity}.
    \end{align}
    The third equation holds because removing and appending systems in incoherent states are in MIO.
    }
    
    \reupdate{On the other hand, from the definition of the ROC of a channel, we have 
    \begin{align}
        &1 + C_R(\cL) \nonumber\\
        =& \min\{\lambda|p \cN_{[0]}+ (1-p)\cK_{[1]}\le \lambda \cM, \cM\in\text{MIO}\}\nonumber\\
        =& \min\left\{\lambda\Big|\cN_{[0]}\le \frac{\lambda \cM - (1-p)\cK_{[1]}}{p}, \cM\in\text{MIO}\right\}\nonumber\\
        =& \min\Big\{\lambda\Big|\cN_{[0]}\le \frac{\lambda+p-1}{p}\cdot \frac{\lambda \cM - (1-p)\cK_{[1]}}{\lambda+p-1},\cM\in\text{MIO}\Big\}\label{eq:theo1:CP-ness}\\
        =& \lambda^*,\label{eq:theo1_optimal_ROC_L}
    \end{align}}
    \reupdate{where we define $\cN_{[0]} = \ketbra{0}{0}\ox\cN$ and $\cM_{[1]} = \ketbra{1}{1}\ox\cM$.  Here, the minimum $\lambda^*$ always exists, which can be straightforwardly derived from Theorem 2 in Ref.~\cite{diaz2018using}.
    For convenience, we define $\cM':=\frac{\lambda \cM - (1-p)\cK_{[1]}}{\lambda+p-1}$ for $\cM \in MIO$.
    Because the minimum is attained by $\lambda \ge 1$, the inequality in Eq.~\eqref{eq:theo1:CP-ness} ensures that $\cM'$ is completely positive, i.e., \reupdate{$0\le [p/(\lambda +p-1) ]\cN_{[0]}\le\cM'$}. Since both $\cM,\cK_{[1]}\in \mathrm{MIO}$, we have $\tr[\cM(\rho)]=\tr\rho=\tr[\cK_{[1]}(\rho)]$ and $\cM(\rho), \cK_{[1]}(\rho) \in \cI$ for any free state $\rho\in\cI$. 
    Therefore, we have }
    $$\cM'(\rho) = \frac{\lambda \cM(\rho)-(1-p)\cK_{[1]}(\rho)}{\lambda +p-1} \in \cI \quad (\forall \rho\in\cI),$$
    and $\tr[\cM'(\rho)]=\tr\rho$.
    Hence, we obtain $\cM'\in \mathrm{MIO}$. 
    Then, \reupdate{noting that $C_R(\cN) = C_R(\cN_{[0]})$}, we have 
    \reupdate{
    \begin{align}
        1 + C_R(\cN) &= 1+C_R(\cN_{[0]})\nonumber\\
        &=\min\{\tau| \cN_{[0]} \le \tau \cM, \cM\in\text{MIO}\},\nonumber\\
        & \le \frac{\lambda^* + p - 1}{p}. \nonumber
    \end{align}
    }
    Because $\lambda^*=1+C_R(\cL)$ from Eq.~\eqref{eq:theo1_optimal_ROC_L}, we arrive at 
    \begin{equation}\label{eq:ROC N le L}
        pC_R(\cN) \le C_R(\cL).
    \end{equation}
    Combined with Eq.~\eqref{eq:theo1_convexity}, we obtain
    \begin{equation}\label{eq:theo1_equality}
        pC_R(\cN) = C_R(\cL).
    \end{equation}
    From Theorem 4 in Ref.~\cite{diaz2018using}, there exists an MIO to simulate channel $\cL$ if and only if $m-1\ge C_R(\cL)$. By considering Eq.~\eqref{eq:theo1_equality}, we can rephrase that there exists an MIO to simulate channel \reupdate{$\cN$} probabilistically if and only if the success probability $p$ satisfies
    \begin{equation}
        \frac{m-1}{C_R(\cN)} \ge p.
    \end{equation}
    Because the probability cannot exceed 1, the maximal success probability is 
    \begin{equation}
        P_{\text{MIO}}(\cN, \Psi_m, \varepsilon=0) = \min \Big\{1, \frac{m-1}{C_R(\cN)}\Big\}.
    \end{equation}
    The proof is complete.
\end{proof}

Theorem~\ref{theo:MIO_simulatable} implies that the success probability of simulating a channel by using MIO and maximally coherent states is always greater than zero, i.e., $P_{\text{MIO}}(\cN, \Psi_m, \varepsilon=0)>0$ with $m\ge 2$. In other words, any quantum channel can be simulated by MIO probabilistically with $\Psi_m$.

According to Theorem~\ref{theo:MIO_simulatable}, if given enough resource, i.e., $C_R(\Psi_m)=m-1 \ge C_R(\cN)$, it implies the success probability equals to 1, thus we can simulate the channel deterministically. Otherwise, if we can only access a limited amount of coherence, i.e., $C_R(\Psi_m)=m-1 < C_R(\cN)$, then the channel fails to be simulated deterministically.
However, we can still succeed in such a simulation probabilistically, with the probability given by the ratio between the robustness of resource state $\Psi_m$ and the robustness of coherent channel $\cN$, i.e.,
\begin{equation}\label{eq:upper_bound_succ_prob}
    p=\frac{m-1}{C_R(\cN)}.
\end{equation}
In this sense, Theorem~\ref{theo:MIO_simulatable} quantitatively shows in terms of the resource measures how the advantage of probabilistic channel simulation over the deterministic one appears. Examples of coherent unitary channels are studied to demonstrate such an advantage in the end of this section.

If we are allowed to simulate target channel $\cN$ up to error $\varepsilon$, it is equivalent to exactly simulate a channel $\cL$ which satisfies $\frac{1}{2}\|\cL-\cN\|_\diamond\le \varepsilon$. From the definition of $\varepsilon$ ROC of a channel, we have $\min_{\cL}C_R(\cL)=C_R^\varepsilon(\cN)$, which leads to the following corollary directly.

\begin{corollary}\label{coro:MIO_approximate_simulation}
    Given an arbitrary target channel $\cN$ and a fixed maximally coherent state $\Psi_m$ of dimension $m\ge 2$, the maximal success probability of channel simulation with MIO up to error $\varepsilon$ is
    \begin{align}
        P_{\rm MIO}(\cN, \Psi_m, \varepsilon) &= \max_{\cL} P_{\rm MIO}(\cL, \Psi_m, 0) \nonumber\\
        &= \min\Big\{1, \frac{m-1}{\min_\cL C_R(\cL)}\Big\}\nonumber\\
        &= \min\Big\{1, \frac{m-1}{C_R^\varepsilon(\cN)}\Big\},
    \end{align}
    \update{where $\cL$ are $\varepsilon$ close to the target channels $\cN$.}
\end{corollary}

From Corollary~\ref{coro:MIO_approximate_simulation}, we can approximately simulate the target channel. The more error we allow, the higher the success probability that we can simulate it.

For a further step, instead of using a maximally coherent state, we consider a general pure coherent state. We obtain the following corollary, which gives a lower bound of the maximal success probability for the general pure resource state.

\begin{corollary}
    Given target channel $\cN$ and coherent pure state $\psi=\proj{\psi}$, where $\ket{\psi}=\sum_{i=1}^n\psi_i\ket{i}$, $\psi_i\neq 0, n\ge 2$, the maximal success probability of channel simulation with MIO up to error $\varepsilon$ is lower bounded by 
    \begin{align}\label{eq:coro_arbitrary_state}
        P_{\rm MIO}(\cN, \psi, \varepsilon) &\ge P_{\rm MIO}(\cN, \Psi_m, \update{\varepsilon}) \times P^{\rm distill}_{\rm MIO}(\proj{\psi}\rightarrow\Psi_m),\\
        &\update{\ge \min\Big\{1, \frac{m-1}{C_R^\varepsilon(\cN)} \times \frac{n^2}{m(\sum_{i=1}^n|\psi_i|^{-2})}\Big\}},
    \end{align}   
    where $P^{\rm distill}_{\rm MIO}(\proj{\psi}\rightarrow\Psi_m)$ is the success probability of coherence distillation with MIO from input state $\proj{\psi}$ to the rank-$m$ maximally coherent state \update{for arbitrary integer $m \ge 2$.}
\end{corollary}

\begin{proof}
    Before starting the proof, we need to note that arbitrary coherent pure state $\ket{\psi}$ can be distilled into maximal coherent state $\Psi_m$ by MIO with non-zero probability~\cite{fang2018probabilistic}, which is
    \begin{equation}\label{eq:MIO_distillation_nonzero}
        P^{\rm distill}_{\rm MIO}(\psi\rightarrow\Psi_m)\ge\frac{n^2}{m(\sum_{i=1}^n|\psi_i|^{-2})}>0,
    \end{equation}
    where $m\ge 2$ is an integer.
    
    In order to simulate channel $\cN$ using a given coherent pure state $\psi$, one feasible method is to probabilistically distill the maximally coherent state $\Psi_m$ from $\psi$ first, then probabilistically implement the target channel using the distilled maximally coherent state. Therefore, the maximal success probability is lower bounded as
    \begin{align}
        P_{\rm MIO}(\cN, \psi, \varepsilon) \ge P_{\rm MIO}(\cN, \Psi_m, \update{\varepsilon}) \times P^{\rm distill}_{\rm MIO}(\proj{\psi}\rightarrow\Psi_m).
    \end{align}
    By considering Eq.~\eqref{eq:MIO_distillation_nonzero} and Corollary~\ref{coro:MIO_approximate_simulation}, we arrive at the lower bound of the maximal success probability as
    \begin{align}
        P_{\rm MIO}(\cN, \psi, \varepsilon) \ge \min\Big\{1, \frac{m-1}{C_R^\varepsilon(\cN)}\Big\} \times \frac{n^2}{m(\sum_{i=1}^n|\psi_i|^{-2})}.
    \end{align}
    The proof is completed.
\end{proof}

This corollary implies that if we take an arbitrary coherent pure state as a resource state, we can simulate an arbitrary channel by MIO with non-zero probability. In other words, any coherent pure state is "useful" in the probabilistic channel simulation with MIO. Even only having a few resources, coherent pure states $\psi$ possess the potential to simulate the target channel, albeit with a small probability of success
$P_{\rm MIO}(\cN, \psi, \varepsilon) > 0$.

So far, we have studied the probabilistic channel simulation with MIO by considering the resource states as maximally coherent states and coherent pure states, respectively. For a more general case, where the given resource state is a coherent mixed state, we provide an efficiently computable SDP to achieve the maximal success probability of channel simulation. Due to the non-linearity of Eq.~\eqref{eq:original_opt_prob_2b}, we cannot formulate Problem~\ref{prob:1} into an SDP directly. 
We thus consider a generalization of the set of free operations $O$, namely the set of free sub normalized quantum operations $O_{sub}$ which consist of completely positive and trace non-increasing maps that are free in the sense of $O$.
By adopting this generalization, we convert the probability optimization into the following expression. A similar technique is also applied in Ref.~\cite{buscemi2017quantum, ishizaka2005multiparticle}.

\begin{lemma}\label{lemma:simplify}
    For any triplet $(\cN, \omega,\varepsilon)$ and operation class $\cO$, the maximal success probability of coherence channel simulation $P_\cO(\cN, \omega,\varepsilon) = \max \{p\in\mathbb{R}_+| \cE(\omega\ox\cdot)=p\cL(\cdot), \frac{1}{2} ||\cL-\cN||_\diamond \le \varepsilon, \cE\in \cO_{\rm sub}\}$, \update{where $\cL$ are the implemented channels}.

\end{lemma}
\begin{proof}
    For any quantum operation $\Lambda(\omega\ox\rho)=\proj{0}\ox\cE_0(\omega\ox\rho) + \proj{1}\ox\cE_1(\omega\ox \rho)$, where $\cE_0$ and $\cE_1$ are two subnormalized operations, we can check that $\Lambda\in\cO$ if and only if $\cE_0, \cE_1\in\cO_{\rm sub}$, and $\cE_0 + \cE_1$ is trace preserving. Thus, the optimization in Eq.~\eqref{eq:original_opt_prob} is equivalent to finding the optimal subnormalized operations $\cE_0$ and $\cE_1$ such that $\cE_0(\omega\ox\rho)=p\cL(\rho), \cE_1(\omega\ox\rho) = (1-p)\mathbb{I}/d, \frac{1}{2}||\cL-\cN||_\diamond\le\varepsilon$, and $\cE_0+\cE_1$ is trace preserving. Since we can take $\cE_1(\omega\ox\rho)=(1-\tr[\cE_0(\omega\ox\rho)])\frac{\mathbb{I}}{d}$, without compromising the success probability, the maximal success probability of coherence channel simulation is only dependent on $\cE_0$, which complete the proof.
\end{proof}

By generalizing the free operation into the subnormalized version, Lemma~\ref{lemma:simplify} simplifies the optimization of the maximal success probability. For further steps, we formulate the success probability optimization as an efficiently computable SDP, which is shown as follows.

\begin{proposition}\label{prop:suc_p_MIO}
    For a given triplet \update{$(\cN_{A\rightarrow B}, \omega, \varepsilon)$}, the maximal success probability to simulate the target channel with MIO is given by $P_{\rm MIO}(\update{\cN_{A\rightarrow B}}, \omega, \varepsilon) = 1/t_{\rm min}$, where $t_{\rm min}$ is given by 
    \begin{subequations}\label{eq:MIO_SDP}
    \begin{align}
        t_{\rm min} = &\min ~ t; \nonumber\\
        {\rm s.t.}& \tr_R[J_{\Tilde{\cE}}(\omega^T\ox\mathbb{I}_A\ox\mathbb{I}_B)] = J_{\cL_{AB}}\label{eq:MIO_SDP_a};\\
        &J_{\Tilde{\cE}} \ge 0, \tr_{B}[J_{\Tilde{\cE}}] \le t \mathbb{I}_{RA}\label{eq:MIO_SDP_b};\\
        &\reupdate{\tr [J_{\Tilde{\cE}}(\ketbra{i,j,k}{i,j,k'}_{RAB})]=0, ~\forall i,j,k\neq k'};\label{eq:MIO_SDP_c}\\
        &\tr_B[J_\cL] = \mathbb{I}_A ;\label{eq:MIO_SDP_d}\\
        &Z\ge 0, Z\ge J_\cL - J_\cN, 
        \tr_B[Z]\le \varepsilon \mathbb{I}_A.\label{eq:MIO_SDP_e}
    \end{align}
    \end{subequations}
\end{proposition}
\begin{proof}
    To get rid of the non-linearity, we consider the map $\Tilde{\cE} = t\cE$, where $t=1/p$, the inverse of success probability, and $\cE$ is a subnormalized MIO. The notations of $J_{\Tilde{\cE}}$ and $J_\cL$ are the \Choi matrix of maps $\Tilde{\cE}$ and $\cL$, respectively. $J_\cN$ is the \Choi matrix of the target channel $\cN$. Eq.~\eqref{eq:MIO_SDP_a} corresponds to the constraint $\Tilde{\cE}(\omega\ox\cdot)=\cL(\cdot)$. Eq.~\eqref{eq:MIO_SDP_b} and Eq.~\eqref{eq:MIO_SDP_c} imply that $\cE$ is a subnormalized MIO. \update{From the result in Ref.~\cite{watrous2009semidefinite} (Sec.4)}, Eq.~\eqref{eq:MIO_SDP_d} and Eq.~\eqref{eq:MIO_SDP_e} guarantee that the simulated channel $\cL$ includes the target channels $\cN$ up to error $\varepsilon$, which is $\frac{1}{2}||\cL-\cN||_\diamond\le\varepsilon$.
\end{proof}

A general qubit unitary possesses four real parameters. However, we can transform unitaries into each other without any additional cost by incorporating coherent unitaries before or after them. This observation implies the existence of an equivalence relation among qubit unitaries up to coherent unitaries. A unique representative of each equivalence class~\cite{dana2017resource,diaz2018using} is given in Eq.~\eqref{eq:rotational_channel}

\begin{equation}\label{eq:rotational_channel}
    U_\theta = \begin{pmatrix}
        \cos \theta & -\sin \theta\\
        \sin \theta & \cos \theta
    \end{pmatrix},
\end{equation}
\update{where $\theta \in [0, \pi/4]$}.

Here we choose the unitary channel $\cU_\theta^{(l)}(\cdot) = U_\theta^{\ox l} \cdot U_\theta^{\dagger\ox l}$, as the target channel to simulate. The success probability of unitary channel simulation $\cU_\theta^{(l)}$ can be derived from the robustness of the channel. We take the default that $\cU_\theta$ refers to $l=1$. The specific statement is shown in the following.

\begin{proposition}\label{prop:unitary_simulation_with_MIO}
    Given the triplet $(\cU_\theta^{(l)}, \Psi_m, \varepsilon=0)$, the maximal success probability of channel simulation is 
    \begin{align}
        P_{\rm MIO}(\cU_\theta^{(l)}, \Psi_m, \varepsilon=0) &= \min\Big\{1, \frac{m-1}{C_R(\cU_\theta^{(l)})}\Big\} \\ 
        &=\min\Big\{1, \frac{m-1}{(1+\sin{2\theta})^{l}-1}\Big\}.
    \end{align}
\end{proposition}
\begin{proof}
    The cohering power~\cite{mani2015cohering, bu2017cohering} of a channel $P_R(\cN)$ is defined as 
    \begin{equation}
        P_R(\cN) = \max_i \log(1+C_R(\cN(\proj{i}))).
    \end{equation}
    We denote $i^*$ to be the optimal solution, i.e., $P_R(\cN)=\log(1+C_R(\cN(\proj{i^*})))$. In~\cite{diaz2018using}, it has been proved that the cohering power of a channel $P_R(\cN)$ is equivalent to the log-robustness of the channel, which is $P_R(\cN) = \log(1+C_R(\cN))$.
    It is straightforward to have 
    \begin{equation}
        C_R(\cN) = C_R(\cN(\proj{i^*})).
    \end{equation}
    
    If we replace the unitary channel $\cU_\theta$ into the equation, we can deduce $C_R(\cU_\theta) = C_R(\cU_\theta(\proj{i^*}))$ directly. For a single-qubit channel, the robustness of a state is equal to its $l_1$ norm of coherence~\cite{napoli2016robustness}, which is the summation of the absolute value of all non-diagonal elements~\cite{baumgratz2014quantifying}. Then, we have
    \begin{equation}
        C_R(\cU_\theta) = C_R(\cU_\theta(\proj{i^*})) = C_{l_1}(\cU_\theta(\proj{i^*}))=\sin2\theta.
    \end{equation}
    Note that the robustness is multiplicative under the tensor product of states~\cite{zhu2017coherence}, specifically
    \begin{equation}
        1 + C_R (\rho_1 \ox \rho_2) = (1 + C_R (\rho_1)) (1 + C_R (\rho_2)).
    \end{equation}
    We correspondingly have 
    \begin{equation}
        C_R(\cU_\theta^{(l)}) = (1+\sin2\theta)^{l}-1.
    \end{equation}
    Recalling Theorem~\ref{theo:MIO_simulatable}, the maximal success probability is the ratio between the robustness of the resource state and the robustness of the target coherent channel, and the success probability $p$ cannot exceed 1. Then we have 
    \begin{equation}
        P_{\rm MIO}(\cU_\theta^{(l)}, \Psi_m, \varepsilon=0) =\min\Big\{1, \frac{m-1}{(1+\sin2\theta)^{(l)}-1}\Big\},
    \end{equation}
    which completes the proof.
\end{proof}

Hence, we obtain the analytic expression of the maximal success probability of the exact simulation of unitary channel $\cU_\theta^{(l)}$ with MIO. We also conduct numerical experiments to show the approximate probabilistic channel simulation. We consider two-qubit unitary channels $\cU_\theta^{l=2}$ $(0 \le \theta\le \pi/4)$ as the target channels. The resource state used is the rank-2 maximally coherent state $\Psi_2$. For different error tolerance $\varepsilon\in\{0, 0.05, 0.1, 0.15, 0.2\}$, the maximal success probabilities of channel simulation with MIO are calculated by the SDP given in Proposition~\ref{prop:suc_p_MIO}. The results are shown in Fig.~\ref{fig:Prob_respect_angle}

\begin{figure}[ht]
\centering
\includegraphics[width=\linewidth]{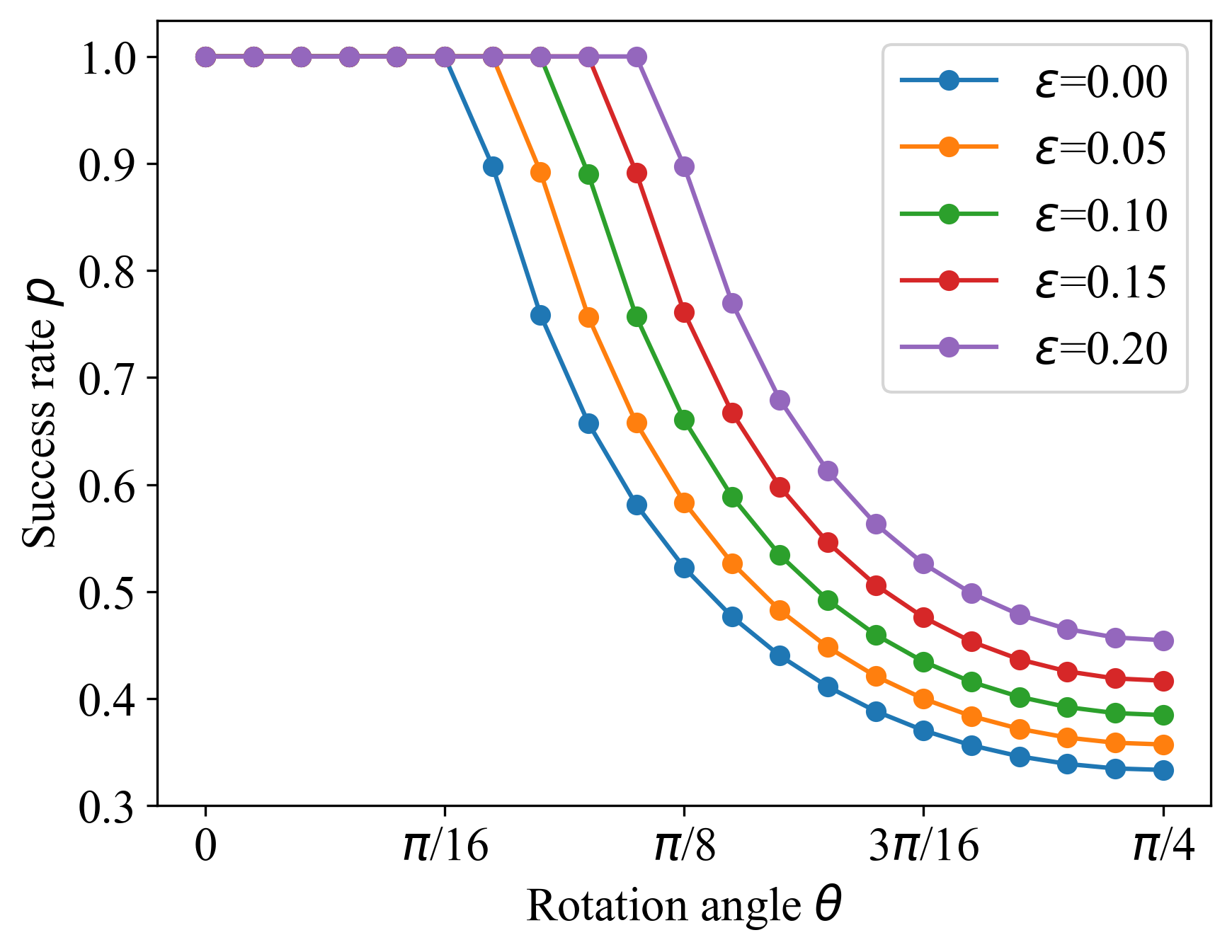}
    \caption{Success probability $p$ of unitary channel simulation with coherent state $\Psi_2$ and MIO. The rotation angle $\theta$ is the parameter of target unitary channel $\cU_\theta^{(2)}$, which is chosen from 0 to $\pi/4$. The resource state is the rank-2 maximally coherent state $\Psi_2$. The five curves from bottom to top correspond to the error tolerance $\varepsilon$ equaling to \{0, 0.05, 0.1, 0.15, 0.2\}, respectively.}
    \label{fig:Prob_respect_angle}
\end{figure}

The robustness of the unitary channel $C_R(\cU_\theta^{(2)})$ increases with respect to the increase of the rotation angle $\theta$. When the angle is small, the resource state $\Psi_2$ has more resources than the required resource of the task, so we can simulate such a channel deterministically. When the rotation angle exceeds some threshold, the coherence of the resource state $\Psi_2$ is not enough for its deterministic implementation. Thus, one can never implement the target channel exactly in the deterministic scenario, and can only realize it approximately, while in the probabilistic scenario, one can exactly implement the target channel at the expense of the reduced success probability.


\section{Probabilistic channel simulation with DIO}

It was proved in Ref.~\cite{theurer2019quantifying} that if the target channel $\cN$ is not a resource nonactivating channel, then it cannot be simulated exactly by DIO with any resource state deterministically. We further show that such an exact simulation is not possible even in probabilistic scenarios. Also, SDP for calculating the success probability of the channel simulation is provided.

In the previous part, we have shown that arbitrary channels can be simulated by using an appropriate resource state and MIO. But for DIO, it is not the same story. Explicitly, we have the following theorem.
\begin{theorem}\label{theo:no-go_theorem}\update{{\rm ~\cite{theurer2019quantifying} (Appendix Eq.(109))}}
If a quantum channel $\cN$ can be implemented by a DIO channel $\cM$ using resource state $\omega$, i.e. $\cM(\omega\ox\cdot) = \cN(\cdot)$, then $\cN$ satisfies $\Delta\circ\cN=\Delta\circ\cN\circ \Delta$.
\end{theorem}

\begin{proof}
    Since $\cN(\rho)=\cM(\omega\ox\cdot)$ holds for any quantum state, we directly have $\Delta\circ\cN(\cdot)=\Delta\circ\cM(\omega\ox\cdot)$. Considering $\cM$ is a DIO channel, which means $\Delta\circ\cM(\omega\ox\cdot) = \cM\circ\Delta(\omega\ox\cdot)$. we arrive at 
    \begin{equation}\label{eq:proof_no_go_1}
        \Delta\circ\cN(\cdot) = \cM\circ\Delta(\omega\ox\cdot).
    \end{equation}
    If the input state is applied to the completely dephasing map, we directly have 
    \begin{align}
        \Delta\circ\cN(\Delta(\cdot)) &= \cM\circ\Delta(\omega\ox\Delta(\cdot))\\
        \Rightarrow\Delta\circ\cN\circ\Delta(\cdot) &=  \cM\circ\Delta(\omega\ox\cdot)\label{eq:proof_no_go_2}.
    \end{align}
    Combine Eq.~\eqref{eq:proof_no_go_1} and Eq.~\eqref{eq:proof_no_go_2}, we have

    \begin{equation}\label{eq:proof_no_go_3}
        \Delta\circ\cN(\cdot) = \Delta\circ\cN\circ\Delta(\cdot)
    \end{equation}
    Eq.~\eqref{eq:proof_no_go_3} holds for any quantum state, which implies $\Delta\circ\cN=\Delta\circ\cN\circ\Delta$, i.e., $\cN$ is a resource nonactivating channel. The proof is complete.
\end{proof}

In other words, this theorem implies that if a quantum channel $\cN$ does not satisfy the condition $\Delta\circ\cN = \Delta\circ\cN\circ\Delta$ (also known as a resource nonactivating channel~\cite{liu2017resource}), then it cannot be simulated by DIO.
Also, it is straightforward to extend this theorem to the probabilistic scenario, which is shown in the following corollary~\ref{coro:no_go_coro}. 

\begin{corollary}~\label{coro:no_go_coro}
    If a quantum channel $\cN$ can be simulated exactly by a DIO channel $\cM$ using resource state $\omega$ with non-zero probability $p$, i.e., $\cM(\omega\ox\cdot) = p\ketbra{0}{0}\ox\cN(\cdot) + (1-p) \ketbra{1}{1}\ox\tr[\cdot]\frac{\mathbb{I}}{d}$, where $\mathbb{I}$ is the identity and $d$ is the dimension, then $\cN$ satisfies $\Delta\circ\cN=\Delta\circ\cN\circ \Delta$. 
\end{corollary}

\begin{proof}
    Denote $\cN'(\cdot) = p\ketbra{0}{0}\ox\cN(\cdot) + (1-p) \ketbra{1}{1}\ox\tr[\cdot]\frac{\mathbb{I}}{d}$. From Theorem~\ref{theo:no-go_theorem}, $\cN'$ is a nonactivating channel, then we directly have $\Delta\circ\cN=\Delta\circ\cN\circ \Delta$, which completes the proof. 
\end{proof}

Theorem~\ref{theo:no-go_theorem} and \update{Corollary}~\ref{coro:no_go_coro} tell us that not all quantum channels can be simulated exactly by DIO, even probabilistically. Take single-qubit unitary channel $\cU_\theta$ and quantum state$\proj{+}$ as an example. One can easily obtain 
\begin{align}
    \Delta\circ\cU_\theta(\update{\proj{+}})&=\begin{pmatrix}
        \frac{1}{2}-\cos\theta\sin\theta & 0\\
        0 & \frac{1}{2}+\cos\theta\sin\theta
    \end{pmatrix}\nonumber\\
    \Delta\circ\cU_\theta\circ\Delta(\update{\proj{+}})&=\begin{pmatrix}
        \frac{1}{2} & 0\\
        0 & \frac{1}{2}
    \end{pmatrix}.\nonumber
\end{align}

$\Delta\circ\cU_\theta\neq \Delta\circ\cU_\theta\circ\Delta$ represents that a single qubit unitary channel $\cU_\theta$ is not a resource nonactivating channel. Thus a qubit unitary channel cannot be exactly simulated by DIO even probabilistically. 

If a quantum channel $\cN$ can be simulated by DIO, then the success probability can be efficiently computed by the SDP as shown in Proposition~\ref{prop:DIO_simulation}.

\begin{proposition}\label{prop:DIO_simulation}
    For a given triplet $(\cN_{A\rightarrow B}, \omega, \varepsilon)$, the maximal success probability to simulate the target channel with DIO is given by $P_{\rm DIO}(\cN_{A\rightarrow B}, \omega, \varepsilon) = 1/t_{\rm min}$, where $t_{\rm min}$ is given by 
    \begin{subequations}\label{eq:DIO_SDP}
    \begin{align}
        t_{\rm min} = &\min ~ t ;\nonumber\\
        {\rm s.t.}& \tr_R[J_{\Tilde{\cE}}(\omega^T\ox\mathbb{I}_A\ox\mathbb{I}_B)] = J_{\cL_{AB}};\label{eq:DIO_SDP_a}\\
        &J_{\Tilde{\cE}} \ge 0, \tr_{B}[J_{\Tilde{\cE}}] \le t \mathbb{I}_{RA};\label{eq:DIO_SDP_b}\\
        &\reupdate{\tr [J_{\Tilde{\cE}}(\ketbra{i,j,k}{m,n,k'}_{RAB})]=0, ~\forall i,j,m,n, k\neq k'};\label{eq:DIO_SDP_d}\\
        &\tr_B[J_\cL] = \mathbb{I}_A ;\label{eq:DIO_SDP_e}\\
        &Z\ge 0, Z\ge J_\cL - J_\cN, 
        \tr_B[Z]\le \varepsilon \mathbb{I}_A;\label{eq:DIO_SDP_f}
    \end{align}
    \end{subequations}
\end{proposition}
Compared with the SDP for success probability of channel simulation with MIO in Eq.~\eqref{eq:MIO_SDP}, the SDP in Eq.~\eqref{eq:DIO_SDP} has 
\reupdate{additional constraints on the matrix components of $J_{\Tilde{\cE}}$ with $(i,j)\neq (m,n)$ as in Eq.~\eqref{eq:DIO_SDP_d}, which corresponds to $\Delta(\cM(\ketbra{i,j}{m,n}_{RA}))=0,~\forall (i,j)\neq (m,n)$}.

\begin{figure}[ht]
\centering
\includegraphics[width=\linewidth]{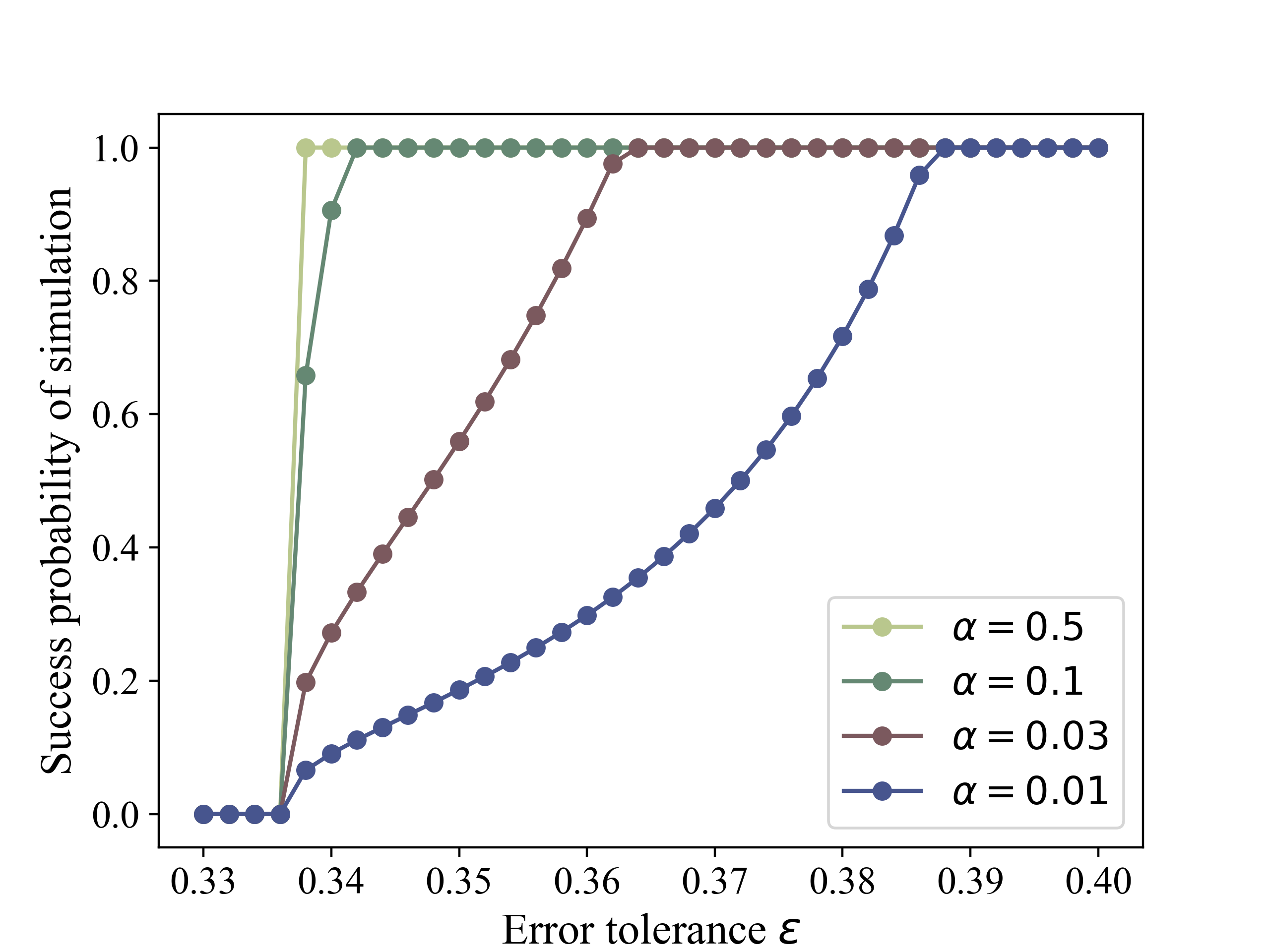}
    \caption{Success probability of simulating a random channel by DIO with different error tolerances. The resource state is pure coherent state $\ket{\psi}=\sqrt{\alpha}\ket{0} + \sqrt{1-\alpha}\ket{0}$, with $\alpha\in\{0.01,0.03,0.1, 0.5\}$.}
    \label{fig:prob_sim_random_channel_unitary_DIO}
\end{figure}

We conduct a numerical experiment to simulate a target channel with DIO. We consider a two-qubit random channel as the target channel, the \Choi matrix of which is shown in Appendix \ref{App:Random_channel}. The resource state is a coherent pure state, $\ket{\psi}=\sqrt{\alpha}\ket{0} + \sqrt{1-\alpha}\ket{1}$ for $\alpha\in [0,0.5]$. The state's coherence increases as parameter $\alpha$ increases. When $\alpha=0.5$, the resource state is a rank-2 maximally coherent state $\Psi_2$. The maximal success probabilities of channel simulation by DIO with different error tolerances and coherent states are calculated by the SDP (Eq.~\eqref{eq:DIO_SDP}).

\update{From the results shown in Fig.~\ref{fig:prob_sim_random_channel_unitary_DIO}, \update{we can see that the probabilistic channel simulation with smaller probability $p$ admits smaller error tolerance.}
In other words, probabilistic relaxation is also effective for DIO. \update{As $\alpha$ increases (larger coherence)}, the success probability of simulation increases faster \update{and gets closer to the case of the maximally coherent state ($\alpha=0.5$)}.}

\section{Conclusion}

We have focused on the probabilistic channel simulation with MIO and DIO, respectively. For the MIO part, we talked about three cases: If the resource state is the maximally coherent state, we provide an analytical expression for the maximal success probability. If we select any pure coherent state as the resource state, the maximal success probability is guaranteed to be greater than zero. If the resource state $\omega$ is a general coherent state, we offer an efficiently computable SDP for achieving the maximal success probability. For the DIO part, we show that not all quantum channels can be exactly simulated by DIO, even probabilistically. Furthermore, we present an efficiently computable SDP for attaining the maximal success probability of channel simulation with DIO.

For further research, it would be interesting to apply the framework of probabilistic channel simulation to other important quantum resources such as entanglement and magic.

\textbf{Acknowledgements.--} B. Zhao would like to thank Xin Wang, Xuanqiang Zhao and Zhiping Liu for their fruitful discussion. This work is supported by Ministry of Education, Culture, Sports, Science, and Technology Quantum Leap Flagship Program Grants No. JPMXS0118067394 and JPMXS0120319794, and JSTJapan Science and Technology Agency COINEXT Grant No. JPMJPF2014. 

\bibliography{references.bib}


\vspace{2cm}
\onecolumngrid
\vspace{2cm}
\begin{center}
{\textbf{\large Appendix for Probabilistic channel simulation of coherence}}
\end{center}

\appendix

\section{Robustness of a state}\label{appen:Robustness of a state}
\begin{definition}{\rm Robustness of state}
    Given an arbitrary quantum state $\rho$, the robustness of a state is defined as ~\cite{napoli2016robustness, piani2016robustness}
    \begin{equation}
        C_R(\rho)=\min_{\tau}\Big\{s\ge 0 \Big| \frac{\rho+s\tau}{1+s}=\sigma \in \cI\Big\},
    \end{equation}
    where $\cI$ refers to the set of free states.
\end{definition}
The robustness of the state quantifies the minimum amount of another state \update{$\tau$ required such that its mixed state $(\rho + s \tau) / (1+s)$ is the incoherent state. We need to note that the added state $\tau$ is not necessarily a free state.} Let us denote the $s^*$ to be the optimal value of $s$, and the corresponding states are denoted as $\tau^*$ and $\sigma^*$. \update{Then} $C_R(\rho)=s^*$, and
\begin{equation}
    \rho = (1+C_R(\rho))\sigma^* - C_R(\rho)\tau^*,
\end{equation}
is said to realize an optimal pseudo mixture for $\rho$. It can be characterized by ~\cite{piani2016robustness}
\begin{equation}
    1 + C_R(\rho) = \min\{\lambda| \rho\le \lambda\sigma, ~\sigma\in\cI\},
\end{equation}
\reupdate{where $\cI$ denotes the set of incoherent states and $\lambda$ is real-valued.} 

\section{SDP for ROC of channel}\label{appen:SDP}
\noindent\textbf{Definition~\ref{def:robustness_of_channel}}
{\rm \cite{diaz2018using}}
    \textit{The robustness of coherence of a quantum channel $\cN$, $C_R(\cN)$, is defined by }
    \begin{equation}\label{eq:app_robustness_of_channel}
        1 + C_R(\cN) := \min\{\lambda:~\cN\le\lambda \cM, \cM \in {\rm MIO}\}.
    \end{equation}

The ROC of a channel can be computed efficiently by SDPs which is shown as follows~\cite{diaz2018using}:
\begin{subequations}\label{eq:app_robustness_2}
    \begin{align}
        1 + C_R(\cN) = \min\quad &  \lambda\nonumber\\
        {\rm s.t.}\quad & 
        J_\cN \le J_{\cM},\label{eq:app_robustness_2a}\\
        &\tr_B[J_{\cM}] = \lambda \mathbb{I}_A,\label{eq:app_robustness_2b}\\
        &\tr_A[J_{\cM} (\proj{i}^T\ox\mathbb{I})] = \Delta(\tr_A[J_{\cM} (\proj{i}^T\ox\mathbb{I})]),\quad \forall \proj{i}\label{eq:app_robustness_2c}.
    \end{align}
\end{subequations}
where $J_\cN$ and $J_{\cM}$ are \Choi matrices of channel $\cN$ and $\cM$, respectively. Eq.~\eqref{eq:app_robustness_2a} corresponds to $\cN\le\lambda\cM$. Note that $\cN$ is completely positive, which implies that $J_{\cM'}\ge J_\cN\ge 0$, and that $\cM'$ is a completely positive map. Eq.~\eqref{eq:app_robustness_2b} implies that the channel $\cM'$ is a trace scaling map. Eq.~\eqref{eq:app_robustness_2c} comes from the definition of MIO, which is $\cM(\rho)\in\cI$, for any free state $\rho$.

\noindent\textbf{Definition~\ref{def:robustness_of_channel_smooth}}{\rm ~\cite{diaz2018using}}
    \textit{The smoothed version of a robustness of channel is called the $\varepsilon$ robustness of coherence of the channel, which is defined by }
    \begin{equation}
        C_R^\varepsilon(\cN) := \min\Big\{C_R(\cL):~\frac{1}{2}\|\cN-\cL\|_\diamond\le\varepsilon\Big\}.
    \end{equation}

The $\varepsilon$-robustness of coherence of a channel can be computed efficiently by SDPs which is shown as follows~\cite{diaz2018using}.
\begin{subequations}
    \begin{align}
        1 + C_R^\varepsilon(\cN) = \min\quad &  \lambda\nonumber\\
        {\rm s.t.}\quad &  
        J_\cM \ge J_\cL\label{eq:app_robustness_4a}\\
        &\tr_B[J_{\cM}] = \lambda \mathbb{I}_A,\label{eq:app_robustness_4b}\\
        &\tr_A[J_{\cM} (\proj{i}^T\ox\mathbb{I})] = \Delta(\tr_A[J_{\cM} (\proj{i}^T\ox\mathbb{I})]),\quad \forall \proj{i}\label{eq:app_robustness_4c},\\
        & V\ge J_\cL - J_\cN \label{eq:app_robustness_4d},\\
        & \tr_B[V] \le \varepsilon \mathbb{I}_A,\label{eq:app_robustness_4e}\\
        & \tr_B[J_\cL] = \mathbb{I}_A,\label{eq:app_robustness_4f}\\
        & J_\cL \ge 0, ~ V\ge 0.\label{eq:app_robustness_4g}
    \end{align}
\end{subequations}

Eq.~\eqref{eq:app_robustness_4a}, Eq.~\eqref{eq:app_robustness_4b} and Eq.~\eqref{eq:app_robustness_4c} are the same as Eq.~\eqref{eq:app_robustness_2}. The rest of the constraints, i.e., Eq.~\eqref{eq:app_robustness_4d}, Eq.~\eqref{eq:app_robustness_4e}, Eq.~\eqref{eq:app_robustness_4f}, Eq.~\eqref{eq:app_robustness_4g} corresponds to $\frac{1}{2}\|\cN-\cL\|_\diamond\le\varepsilon$.

\section{Random channel simulation with DIO}\label{App:Random_channel}
The target channel $\cN$ used in the numerical experiment of channel simulation with DIO (Fig.~\ref{fig:prob_sim_random_channel_unitary_DIO}) is a random two-qubit rank-4 real channel, of which the \Choi matrix $J_\cN$ is 

\begin{align}
    &J_\cN=\nonumber\\
    &\begin{pmatrix}
    \begin{smallmatrix}
    0.1425&-0.0717&0.0663&-0.0582&0.1640&0.0037&-0.0010&0.0504&0.0095&0.0560&0.0765&-0.1369&-0.0963&-0.1184&0.0484&-0.0148\\
   -0.0717&0.2731&-0.0653&0.0888&-0.1178&-0.0300&0.0830&-0.2387&0.2253&-0.1143&0.0338&-0.0622&0.1673&0.2035&-0.3014&-0.1055\\
    0.0663&-0.0653&0.2067&0.0887&0.1565&-0.0341&-0.0497&0.0280&0.1282&-0.1035&0.0352&0.0825&-0.1642&-0.1078&0.0018&0.0301\\
   -0.0582&0.0888&0.0887&0.3776&-0.0611&-0.2899&-0.2548&-0.0843&0.1231&-0.1234&0.1396&0.0695&-0.1340&-0.0933&-0.1351&-0.1089\\
    0.1640&-0.1178&0.1565&-0.0611&0.2334&0.0313&0.0094&0.0779&0.0515&0.0115&0.0580&-0.0741&-0.1535&-0.1478&0.0720&0.0257\\
    0.0037&-0.0300&-0.0341&-0.2899&0.0313&0.2686&0.2506&0.0248&-0.0198&0.0210&-0.1566&0.0291&0.1374&0.1432&0.0514&0.1116\\
   -0.0010&0.0830&-0.0497&-0.2548&0.0094&0.2506&0.2845&-0.0765&0.0879&-0.0219&-0.1227&-0.0273&0.1929&0.2113&-0.0801&0.0578\\
    0.0504&-0.2387&0.0280&-0.0843&0.0779&0.0248&-0.0765&0.2135&-0.2269&0.1179&-0.0351&0.0502&-0.1302&-0.1712&0.2746&0.0904\\
    0.0095&0.2253&0.1282&0.1231&0.0515&-0.0198&0.0879&-0.2269&0.3708&-0.2111&0.0589&-0.0127&0.0375&0.1278&-0.3174&-0.0736\\
    0.0560&-0.1143&-0.1035&-0.1234&0.0115&0.0210&-0.0219&0.1179&-0.2111&0.1715&0.0095&-0.1161&-0.0078&-0.0855&0.1671&0.0079\\
    0.0765&0.0338&0.0352&0.1396&0.0580&-0.1566&-0.1227&-0.0351&0.0589&0.0095&0.1489&-0.1334&-0.0962&-0.1113&-0.0639&-0.1010\\
   -0.1369&-0.0622&0.0825&0.0695&-0.0741&0.0291&-0.0273&0.0502&-0.0127&-0.1161&-0.1334&0.3088&-0.0268&0.0376&0.0640&0.1119\\
   -0.0963&0.1673&-0.1642&-0.1340&-0.1535&0.1374&0.1929&-0.1302&0.0375&-0.0078&-0.0962&-0.0268&0.2532&0.2514&-0.1346&-0.0071\\
   -0.1184&0.2035&-0.1078&-0.0933&-0.1478&0.1432&0.2113&-0.1712&0.1278&-0.0855&-0.1113&0.0376&0.2514&0.2868&-0.1947&0.0025\\
    0.0484&-0.3014&0.0018&-0.1351&0.0720&0.0514&-0.0801&0.2746&-0.3174&0.1671&-0.0639&0.0640&-0.1346&-0.1947&0.3599&0.1212\\
   -0.0148&-0.1055&0.0301&-0.1089&0.0257&0.1116&0.0578&0.0904&-0.0736&0.0079&-0.1010&0.1119&-0.0071&0.0025&0.1212&0.1001\\
    \end{smallmatrix}
\end{pmatrix}.\nonumber
\end{align}

We can easily check that $\Delta\circ\cN \neq \Delta\circ\cN\circ\Delta$, implying that the target channel $\cN$ cannot be exactly simulated by DIO both deterministically and probabilistically.
\end{document}

%% file: pretex.tex
\usepackage{mathtools}
\usepackage{amsmath}
\usepackage[shortlabels]{enumitem}

\usepackage{graphicx,epic,eepic,epsfig,amsmath,latexsym,amssymb,verbatim,color}
 
\usepackage{amsfonts}       
\usepackage{nicefrac}       

\usepackage{amsmath}
\usepackage{bbm}

\usepackage{float}
\usepackage{tikz}
\usetikzlibrary{chains}
\usetikzlibrary{fit}
\usetikzlibrary{arrows}

\usepackage{epsfig}
\usetikzlibrary{shapes.symbols,patterns} 
\usepackage{pgfplots}
\pgfplotsset{compat=1.18}

\usepackage[strict]{changepage}
\usepackage{hyperref}
\hypersetup{colorlinks=true,citecolor=blue,linkcolor=blue,filecolor=blue,urlcolor=blue,breaklinks=true}

\usepackage[marginal]{footmisc}
\usepackage{url}
\usepackage{theorem}

\newtheorem{definition}{Definition}
\newtheorem{proposition}{Proposition}
\newtheorem{lemma}[proposition]{Lemma}

\newtheorem{theorem}[proposition]{Theorem}

\newtheorem{corollary}[proposition]{Corollary}

\def\squareforqed{\hbox{\rlap{$\sqcap$}$\sqcup$}}
\def\qed{\ifmmode\squareforqed\else{\unskip\nobreak\hfil
\penalty50\hskip1em\null\nobreak\hfil\squareforqed
\parfillskip=0pt\finalhyphendemerits=0\endgraf}\fi}
\def\endenv{\ifmmode\;\else{\unskip\nobreak\hfil
\penalty50\hskip1em\null\nobreak\hfil\;
\parfillskip=0pt\finalhyphendemerits=0\endgraf}\fi}
\newenvironment{proof}{\noindent \textbf{{Proof~} }}{\hfill $\blacksquare$}

\newcounter{remark}

\newcounter{example}

\mathchardef\ordinarycolon\mathcode`\:
\mathcode`\:=\string"8000
\def\vcentcolon{\mathrel{\mathop\ordinarycolon}}
\begingroup \catcode`\:=\active
  \lowercase{\endgroup
  \let :\vcentcolon
  }

\usepackage{cleveref}
\usepackage{graphicx}
\usepackage{xcolor}

\RequirePackage[framemethod=default]{mdframed}
\newmdenv[skipabove=7pt,
skipbelow=7pt,
backgroundcolor=darkblue!15,
innerleftmargin=5pt,
innerrightmargin=5pt,
innertopmargin=5pt,
leftmargin=0cm,
rightmargin=0cm,
innerbottommargin=5pt,
linewidth=1pt]{tBox}

\newmdenv[skipabove=7pt,
skipbelow=7pt,
backgroundcolor=blue2!25,
innerleftmargin=5pt,
innerrightmargin=5pt,
innertopmargin=5pt,
leftmargin=0cm,
rightmargin=0cm,
innerbottommargin=5pt,
linewidth=1pt]{dBox}
\newmdenv[skipabove=7pt,
skipbelow=7pt,
backgroundcolor=darkkblue!15,
innerleftmargin=5pt,
innerrightmargin=5pt,
innertopmargin=5pt,
leftmargin=0cm,
rightmargin=0cm,
innerbottommargin=5pt,
linewidth=1pt]{sBox}
\definecolor{darkblue}{RGB}{0,76,156}
\definecolor{darkkblue}{RGB}{0,0,153}
\definecolor{blue2}{RGB}{102,178,255}
\definecolor{darkred}{RGB}{195,0,0}

\newcommand{\nc}{\newcommand}
\nc{\rnc}{\renewcommand}
\nc{\lbar}[1]{\overline{#1}}
\nc{\bra}[1]{\langle#1|}
\nc{\ket}[1]{|#1\rangle}
\nc{\ketbra}[2]{|#1\rangle\!\langle#2|}
\nc{\braket}[2]{\langle#1|#2\rangle}

\nc{\proj}[1]{| #1\rangle\!\langle #1 |}
\nc{\avg}[1]{\langle#1\rangle}
\nc{\rank}{\operatorname{Rank}}
\nc{\smfrac}[2]{\mbox{$\frac{#1}{#2}$}}
\nc{\tr}{\operatorname{Tr}}
\nc{\ox}{\otimes}
\nc{\dg}{\dagger}
\nc{\dn}{\downarrow}
\nc{\cA}{{\cal A}}
\nc{\cB}{{\cal B}}
\nc{\cC}{{\cal C}}
\nc{\cD}{{\cal D}}
\nc{\cE}{{\cal E}}
\nc{\cF}{{\cal F}}
\nc{\cG}{{\cal G}}
\nc{\cH}{{\cal H}}
\nc{\cI}{{\cal I}}
\nc{\cJ}{{\cal J}}
\nc{\cK}{{\cal K}}
\nc{\cL}{{\cal L}}
\nc{\cM}{{\cal M}}
\nc{\cN}{{\cal N}}
\nc{\cO}{{\cal O}}
\nc{\cP}{{\cal P}}
\nc{\cQ}{{\cal Q}}
\nc{\cR}{{\cal R}}
\nc{\cS}{{\cal S}}
\nc{\cT}{{\cal T}}
\nc{\cU}{{\cal U}}
\nc{\cV}{{\cal V}}
\nc{\cX}{{\cal X}}
\nc{\cY}{{\cal Y}}
\nc{\cZ}{{\cal Z}}
\nc{\cW}{{\cal W}}
\nc{\csupp}{{\operatorname{csupp}}}
\nc{\qsupp}{{\operatorname{qsupp}}}
\nc{\var}{{\operatorname{var}}}
\nc{\rar}{\rightarrow}
\nc{\lrar}{\longrightarrow}
\nc{\polylog}{{\operatorname{polylog}}}
\nc{\wt}{{\operatorname{wt}}}
\nc{\av}[1]{{\left\langle {#1} \right\rangle}}
\nc{\supp}{{\operatorname{supp}}}

\nc{\argmin}{{\operatorname{argmin}}}

\def\i{\mathbf{i}}

\def\x{\xi}

\nc{\RR}{{{\mathbb R}}}
\nc{\CC}{{{\mathbb C}}}
\nc{\FF}{{{\mathbb F}}}
\nc{\NN}{{{\mathbb N}}}
\nc{\ZZ}{{{\mathbb Z}}}
\nc{\PP}{{{\mathbb P}}}
\nc{\QQ}{{{\mathbb Q}}}
\nc{\UU}{{{\mathbb U}}}
\nc{\EE}{{{\mathbb E}}}
\nc{\id}{{\operatorname{id}}}

\nc{\CHSH}{{\operatorname{CHSH}}}

\nc{\be}{\begin{equation}}
\nc{\ee}{{\end{equation}}}
\nc{\bea}{\begin{eqnarray}}
\nc{\eea}{\end{eqnarray}}
\nc{\<}{\langle}
\rnc{\>}{\rangle}
\nc{\rU}{\mbox{U}}

\nc{\ob}[1]{#1}

\nc{\SEP}{{\text{\rm SEP}}}
\nc{\NS}{{\text{\rm NS}}}
\nc{\LOCC}{{\text{\rm LOCC}}}
\nc{\PPT}{{\text{\rm PPT}}}
\nc{\EXT}{{\text{\rm EXT}}}
\nc{\Sym}{{\operatorname{Sym}}}

\nc{\ERLO}{{E_{\text{r,LO}}}}
\nc{\ERLOCC}{{E_{\text{r,LOCC}}}}
\nc{\ERPPT}{{E_{\text{r,PPT}}}}
\nc{\ERLOCCinfty}{{E^{\infty}_{\text{r,LOCC}}}}
\nc{\Aram}{{\operatorname{\sf A}}}
\newtheorem{problem}{Problem}

\newcommand{\Choi}{Choi-Jamio\l{}kowski }

\usepackage{tikz}
\usepackage{hyperref}
\hypersetup{colorlinks=true,citecolor=blue,linkcolor=blue,filecolor=blue,urlcolor=blue,breaklinks=true}

\makeatletter
\def\grd@save@target#1{%
  \def\grd@target{#1}}
\def\grd@save@start#1{%
  \def\grd@start{#1}}
\tikzset{
  grid with coordinates/.style={
    to path={%
      \pgfextra{%
        \edef\grd@@target{(\tikztotarget)}%
        \tikz@scan@one@point\grd@save@target\grd@@target\relax
        \edef\grd@@start{(\tikztostart)}%
        \tikz@scan@one@point\grd@save@start\grd@@start\relax
        \draw[minor help lines,magenta] (\tikztostart) grid (\tikztotarget);
        \draw[major help lines] (\tikztostart) grid (\tikztotarget);
        \grd@start
        \pgfmathsetmacro{\grd@xa}{\the\pgf@x/1cm}
        \pgfmathsetmacro{\grd@ya}{\the\pgf@y/1cm}
        \grd@target
        \pgfmathsetmacro{\grd@xb}{\the\pgf@x/1cm}
        \pgfmathsetmacro{\grd@yb}{\the\pgf@y/1cm}
        \pgfmathsetmacro{\grd@xc}{\grd@xa + \pgfkeysvalueof{/tikz/grid with coordinates/major step}}
        \pgfmathsetmacro{\grd@yc}{\grd@ya + \pgfkeysvalueof{/tikz/grid with coordinates/major step}}
        \foreach \x in {\grd@xa,\grd@xc,...,\grd@xb}
        \node[anchor=north] at (\x,\grd@ya) {\pgfmathprintnumber{\x}};
        \foreach \y in {\grd@ya,\grd@yc,...,\grd@yb}
        \node[anchor=east] at (\grd@xa,\y) {\pgfmathprintnumber{\y}};
      }
    }
  },
  minor help lines/.style={
    help lines,
    step=\pgfkeysvalueof{/tikz/grid with coordinates/minor step}
  },
  major help lines/.style={
    help lines,
    line width=\pgfkeysvalueof{/tikz/grid with coordinates/major line width},
    step=\pgfkeysvalueof{/tikz/grid with coordinates/major step}
  },
  grid with coordinates/.cd,
  minor step/.initial=.2,
  major step/.initial=1,
  major line width/.initial=2pt,
}
\makeatother

\usepackage{thmtools}
\usepackage{thm-restate}
\usepackage{etoolbox}
\makeatletter
\def\problem@s{}
\newcounter{problems@cnt}

\newcommand{\allproblems}{\problem@s}
\makeatother